\documentclass[12pt,reqno]{amsart}
\usepackage{amsthm,amsfonts,amssymb,euscript}

\theoremstyle{plain}
  \newtheorem{theorem}[subsection]{Theorem}
  
  \newtheorem{proposition}[subsection]{Proposition}
  \newtheorem{lemma}[subsection]{Lemma}
  \newtheorem{corollary}[subsection]{Corollary}
\def\bphi{{\boldsymbol{\phi} }}
\def\bpsi{{\boldsymbol{\psi} }}
\def\F{{\mathcal F}}
\def\R{{\mathbb{R}}}
\def\C{{\mathbb{C}}}
\def\I{{\mathcal I}}
\def\Q{{\mathcal Q}}
\def\ch{\mbox{ch} (k)}
\def\chbb{\mbox{chb} (k)}
\def\trigh{\mbox{trigh} (k)}
\def\p{\mbox{p} (k)}
\def\sh{\mbox{sh} (k)}
\def\shh{\mbox{sh}}
\def\chh{\mbox{ch}}
\def\shhb{\overline {\mbox{sh}}}
\def\chhb{\overline {\mbox{ch}}}
\def\shbb{\mbox{shb} (k)}
\def\chb{\overline{\mbox{ch}  (k)}}
\def\shb{\overline{\mbox{sh}(k)}}
\theoremstyle{remark}
  \newtheorem{remark}[subsection]{Remark}
  
\theoremstyle{definition}

\begin{document}

\def\bphi{{\boldsymbol{\phi} }}
\def\bpsi{{\boldsymbol{\psi} }}
\def\F{{\mathcal F}}
\def\R{{\mathbb{R}}}
\def\C{{\mathbb{C}}}
\def\I{{\mathcal I}}
\def\Q{{\mathcal Q}}
\def\ch{\mbox{\rm ch} (k)}
\def\chbb{\mbox{\rm chb} (k)}
\def\trigh{\mbox{\rm trigh} (k)}
\def\p{\mbox{p} (k)}
\def\sh{\mbox{\rm sh} (k)}
\def\shh{\mbox{\rm sh}}
\def\chh{\mbox{ \rm ch}}
\def\shhb{\overline {\mbox{\rm sh}}}
\def\chhb{\overline {\mbox{\rm ch}}}
\def\shbb{\mbox{\rm shb} (k)}
\def\chb{\overline{\mbox{\rm ch}  (k)}}
\def\shb{\overline{\mbox{\rm sh}(k)}}

\title[Second order corrections
, I]%
{Second order corrections\\ to mean field evolution
for weakly interacting Bosons, I}

\author{M. Grillakis}
\address{University of Maryland, College Park}
\email{mng@math.umd.edu}

\author{M. Machedon}
\address{University of Maryland, College Park}
\email{mxm@math.umd.edu}

\author{D. Margetis}
\address{University of Maryland, College Park}
\email{dio@math.umd.edu}

\thanks{
The first two authors thank William Goldman and John Millson for
discussions related to the Lie algebra of the symplectic group, and
Sergiu Klainerman for the interest shown for this work.
The third author is grateful to Tai Tsun Wu for useful
discussions on the physics of the Boson system.
The third author's research was partially supported
by the NSF-MRSEC grant DMR-0520471 at the University of Maryland,
and by the Maryland NanoCenter.}

\date{}
\dedicatory{}
\commby{}
\maketitle

\maketitle
\begin{abstract}
Inspired by the works of Rodnianski and Schlein \cite{Rod-S} and Wu \cite{wuI,wuII}, we derive
a new nonlinear Schr\"odinger equation that describes
a second-order correction to the usual tensor product (mean-field) approximation for the Hamiltonian evolution of a
many-particle system in Bose-Einstein condensation. We show that our
new equation, if it has solutions with appropriate smoothness and decay properties, implies a new Fock space estimate.
We also show that for an interaction potential $v(x)= \epsilon \chi(x) |x|^{-1}$, where $\epsilon$ is sufficiently small and
$\chi \in C_0^{\infty}$  even, our program can be easily implemented locally in time. We leave global in time issues,
more singular potentials and sophisticated estimates for a subsequent part (part II) of this paper.
\end{abstract}

\section{Introduction}
\label{sec:intro}

An advance in physics in 1995 was the first experimental observation of atoms with integer spin (Bosons)
occupying a macroscopic quantum state (condensate) in a dilute gas at very low temperatures~\cite{andersonetal95,davisetal95}.
This phenomenon of Bose-Einstein condensation has been observed in many
similar experiments since.  These observations have rekindled interest
in the quantum theory of large Boson systems. For recent reviews, see
e.g.~\cite{lieb05,pitaevskii03}.

A system of $N$ interacting Bosons at zero temperature
is described by a symmetric wave function satisfying the $N$-body Schr\"odinger equation.
For large $N$, this description is impractical. It is thus desirable to
replace the many-body evolution by effective (in an appropriate sense) partial differential
equations for wave functions in much lower space dimensions. This approach has led to  ``mean-field" approximations
in which the single particle wave function for the condensate satisfies
nonlinear Schr\"odinger equations (in $3+1$ dimensions). Under this approximation,
the $N$-body wave function is viewed simply as a tensor product of one-particle states.
For early related works,
see the papers by Gross~\cite{gross61,gross63}, Pitaevskii~\cite{pitaevskii61} and Wu~\cite{wuI,wuII}.
In particular, Wu~\cite{wuI,wuII} introduced a {\it second-order approximation} for the Boson many-body wave function
in terms of the {\it pair-excitation} function, a suitable kernel that describes the scattering
of atom {\it pairs} from the condensate to other states. Wu's formulation forms a nontrivial extension of
works by Lee, Huang and Yang~\cite{Lee-H-Y} for the periodic Boson system. Approximations
carried out for pair excitations~\cite{Lee-H-Y,wuI,wuII} make use of quantized fields in the Fock space.
(The Fock space formalism and Wu's formulation are reviewed in sections~\ref{subsec:fock} and~\ref{subsec:wu-pair}, respectively.)

Connecting mean-field approaches to the actual many-particle
Hamiltonian evolution raises fundamental questions. One question is the rigorous derivation and interpretation
of the mean field limit. Elgart, Erd\H{o}s, Schlein and Yau~\cite{E-E-S-Y1,E-Y1,E-S-Y1,E-S-Y2,E-S-Y4,E-S-Y3} showed rigorously
how mean-field limits for Bosons can be extracted in the limit $N\to\infty$ by using Bogoliubov-Born-Green-Kirkwood-Yvon (BBGKY) hierarchies for reduced density matrices.
Another issue concerns the convergence of the microscopic evolution towards the mean field dynamics.
Recently, Rodnianski and Schlein~\cite{Rod-S} provided estimates for the rate of convergence in the
case with Hartree dynamics by invoking the formalism of Fock space.

In this paper, inspired by the works of Rodnianski and Schlein~\cite{Rod-S} and Wu~\cite{wuI,wuII},
we derive a new nonlinear Schr\"odinger equation describing an improved approximation for the evolution of the Boson system.
This approximation offers a second-order correction to the usual tensor product (mean field limit)
for the many-body wave function. Our equation yields a corresponding new estimate in Fock space, which
complements nicely the previous estimate~\cite{Rod-S}.

The static version of the many-body problem is not studied here. The energy spectrum  was addressed by
Dyson~\cite{D} and by Lee, Huang and Yang~\cite{Lee-H-Y}. A mathematical proof of the Bose-Einstein condensation
for the time-independent case was provided recently by Lieb, Seiringer, Solovej and Yngvanson~\cite{LS,lieb05,LSY1,LSY2}.

\subsection{Fock space formalism}
\label{subsec:fock}

Next, we review the Fock space $\F$ over $L^2(\R^3)$,
following Rodnianski and Schlein~\cite{Rod-S}.
The elements of $ \F$ are
vectors of the form $\bpsi=(\psi_0, \psi_1(x_1), \psi_2(x_1, x_2), \cdots )$,
where $\psi_0 \in \C$ and $\psi_n \in L^2_s (\R^{3n})$ are symmetric in $x_1, \ldots, x_n$. The Hilbert space structure of $\F$
is given by $\left(\bphi, \bpsi\right)= \sum_n \int \phi_n \overline{\psi_n} dx$.

For $f \in L^2(\R^3)$ the (unbounded, closed,  densely defined) creation operator
$a^* (f) :\F \to \F$ and annihilation operator $a(\bar f) : \F \to \F$ are defined by

\begin{align*}
&\left(a^*(f)\psi_{n-1}\right)(x_1, x_2, \cdots, x_n)=\frac{1}{\sqrt n}
\sum_{j=1}^n f(x_j)\psi_{n-1}(x_1, \cdots, x_{j-1}, x_{j+1}, \cdots x_n)~,\\
&\left(a(\overline f)\psi_{n+1}\right)(x_1, x_2, \cdots, x_n)=\sqrt{ n+1}
\int \psi_{(n+1)}(x, x_1, \cdots, x_n) \overline f(x)\ dx~.
\end{align*}
The operator valued distributions $a^*_x$ and $a_x$ defined by
\begin{align*}
a^*(f) = \int f(x) a^*_x\ dx~,\\
a(\overline f) = \int \overline f(x)\,a_x\ dx~.
\end{align*}
These distributions satisfy the canonical commutation relations
\begin{align}
[a_x, a^*_y] = \delta(x-y)~,\label{canonical} \\
[a_x, a_y] = [a^*_x, a^*_y] =0 \notag~.
\end{align}

Let $N$ be a fixed integer (the total number of particles), and $v(x)$ be an even potential.
Consider the Fock space Hamiltonian $H_N: \F \to \F$ defined by
\begin{align} \label{H}
H_N&= \int a_x^* \Delta a_x dx + \frac{1}{2N} \int v(x-y) a_x^* a_y^* a_x a_y\ dx\, dy\\
&=: H_0 +\frac{1}{N} V \notag~.
\end{align}
This $H_N$ is a diagonal operator which acts on each $\psi_n$ in correspondence to the
 Hamiltonian
\begin{equation*}H_{N, n } = \sum_{j=1}^{n} \Delta_{x_j} + \frac{1}{2N} \sum_{i\neq j} v(x_i-x_j)~.
\end{equation*}
In the particular case $n=N$, this is the mean field Hamiltonian.
Except for the introduction, this paper deals only with the Fock space Hamiltonian. The reader is alerted that ``PDE" Hamiltonians such as
$H_{N, n}$ will always have two subscripts.
The sign of $v$ will not play a role in our analysis. However, the reader is alerted that due to our sign convention, $v \le 0$ is the "good" sign.
The  time evolution in the coordinate space for Bose-Einstein condensation deals with the function
\begin{equation}e^{i t H_{n, n}} \psi_0 \label{classical}
\end{equation}
for tensor product initial data, i.e., if
\begin{equation*}
\psi_0(x_1, x_2, \cdots, x_n)= \phi_0(x_1) \phi_0(x_2) \cdots \phi_0(x_n)~,
\end{equation*}
where $\|\phi_0\|_{L^2(\R^3)}=1$. This  approach has been highly successful, even for very
singular potentials, in the work of Elgart, Erd\H{o}s, Schlein and Yau~\cite{E-E-S-Y1,E-Y1,E-S-Y1,E-S-Y2,E-S-Y4,E-S-Y3}.
In this context, the convergence of evolution to the appropriate mean field limit (tensor product)
as $N\to \infty$ is established at the level of marginal density matrices
$\gamma_i^{(N)}$ in the trace norm topology. The density matrices are defined as
\begin{align*}
\gamma_i^{(N)}(t, x_1, \cdots, x_i; x'_1, \cdots x'_i)
=\int \psi(t, x_1, \cdots, x_N) \overline{\psi}(t, x'_1, \cdots, x'_N) dx_{i+1} \cdots dx_N
\end{align*}

\subsection{Coherent states}
\label{subsec:coh}

There are alternative
approaches, due to Hepp \cite{hepp}, Ginibre and Velo \cite{G-V}, and, most recently, Rodnianski and Schlein \cite{Rod-S}
which can treat Coulomb potentials $v$.
These approaches rely on studying the Fock space evolution $e^{it H_N} \bpsi_0$ where the initial data $\bpsi_0$
is a coherent state,
\begin{align*}
\bpsi_0=(c_0, c_1 \phi_0(x_1), c_2 \phi_0(x_1) \phi_0(x_2), \cdots)~;
\end{align*}
see \eqref{coherent} below. The evolution \eqref{classical} can then be extracted as a ``Fourier coefficient''
from the Fock space evolution; see \cite{Rod-S}. Under the assumption that $v$ is a Coulomb potential, this approach
leads to strong $L^2$-convergence, still at the level of
the density matrices $\gamma_i^{(N)}$, as we will briefly explain below.

To clarify the issues involved, let us consider the one-particle wave function
$\phi(t, x)$ (to be determined later as the solution of a Hartree equation), satisfying the initial condition
$\phi(0, x) = \phi_0(x)$. Define the skew-Hermitian unbounded operator
\begin{align*}
A(\phi)=a( \overline \phi) - a^*( \phi)
\end{align*}
and the vacuum state $\Omega =(1, 0, 0, \cdots) \in \F$. Accordingly, consider the operator
\begin{align*}
W(\phi) = e^{- \sqrt N A(\phi)}~,
\end{align*}
which is the Weyl operator used by Rodnianski and Schlein~\cite{Rod-S}. The coherent state for the initial data $\phi_0$ is
\begin{align} \notag
\bpsi_0 &=  W(\phi_0)\Omega=e^{- \sqrt N A(\phi_0)} \Omega\\
&= e^{-N \|\phi\|^2/2}\left(1, \cdots, \left(\frac{N^n}{ n!} \right)^{1/2}\phi_0(x_1) \cdots \phi_0(x_n), \cdots \right)~. \label{coherent}
\end{align}
Hence, the top candidate approximation for $ e^{itH_N} \bpsi_0$ reads
\begin{align}
\bpsi_{\mbox{tensor}}(t) =  e^{- \sqrt N A(\phi(t, \cdot))} \Omega~.
\label{firstapprox}
\end{align}
Rodnianski and Schlein~\cite{Rod-S} showed that this approximation works (under suitable assumptions on $v$),
in the sense that
\begin{align*}
&\frac{1}{N}
\|
\left( e^{itH_N} \bpsi_0, \,  a_y^* a_x  e^{itH_N} \bpsi_0 \right)
- \left(e^{- \sqrt N A(\phi(t, \cdot))} \Omega, \,
a^*_y a_x e^{- \sqrt N A(\phi(t, \cdot))}
\Omega  \right)\|_{\mbox{Tr}} \\
&= O(\frac{e^{ C t}}{N})\qquad N\to\infty~;
\end{align*}
the symbol Tr here stands for the trace norm in $x \in \R^3$ and $y \in \R^3$.
The first term in the last relation, including $\frac{1}{N}$, is essentially the density matrix
$\gamma_1^{(N)}(t, x, y)$. For the precise statement of the problem and details of the proof, see
Theorem 3.1 of Rodnianski and Schlein \cite{Rod-S}.

Our goal here is to find  an  explicit approximation for the evolution in the Fock space.
For this purpose, we adopt an idea germane to Wu's
second-order approximation for the $N$-body wave function in Fock space~\cite{wuI,wuII}.

\subsection{Wu's approach}
\label{subsec:wu-pair}

We first comment on the case with periodic boundary conditions, when the condensate is the zero-momentum state.
For this setting, Lee, Huang and Yang~\cite{Lee-H-Y} studied systematically
the scattering of atoms from the condensate to states of opposite momenta. By diagonalizing an
approximation for the Hamiltonian in Fock space, these authors derived a formula for the $N$-particle wave function
that deviates from the usual tensor product, as it expresses excitation of particles from zero monentum to {\it pairs}
of opposite momenta.

For non-periodic settings, Wu~\cite{wuI,wuII} invokes the splitting $a_x=a_0(t)\phi(t,x)+a_{x,1}(t)$
where $a_0$ corresponds to the condensate, $[a_0, a_0^*]=1$, and $a_{x,1}$ corresponds to states
orthogonal to the condensate, $[a_0,a_{x,1}]=0=[a_0,a_{x,1}^*]$.
Wu applies the following ansatz for the $N$-body wave function in Fock space:
\begin{equation}
\mathcal N(t)\,e^{\mathcal P[K_0]}\psi^0_N(t)~,
\label{eq:pair-Psi}
\end{equation}
where $\psi^0_N(t)$ describes the tensor product, $\mathcal N(t)$ is a normalization factor,
and $\mathcal P[K_0]$ is an operator that averages out in space the excitation of particles from the condensate $\phi$ to other states
with the effective kernel (pair excitation function) $K_0$. An explicit formula for $\mathcal P[K_0]$ is
\begin{equation}
\mathcal P[K_0]=[2N_0(t)]^{-1}\int a^*_{x,1}a^*_{y,1}\,K_0(t,x,y)\,a_0(t)^2~,
\end{equation}
where $N_0$ is the expectation value of particle number at the condensate.
This $K_0$ is not a-priori known (in contrast
to the case of the classical Boltzmann gas) but is determined by means consistent with the many-body dynamics.
In the periodic case, \eqref{eq:pair-Psi}
reduces to the many-body wave function of Lee, Huang and Yang~\cite{Lee-H-Y}.

Wu derives a coupled system of dispersive hyperbolic partial differential equations
for $(\phi, K_0)$ via an approximation for the $N$-body Hamiltonian
that is consistent with ansatz~\eqref{eq:pair-Psi}. A feature of this system is the {\it spatially nonlocal} couplings induced by $K_0$.
Observable quantities such that the depletion of the condensate can be
computed directly from solutions of this PDE system. This system has been solved
only in a limited number of cases~\cite{wuII,MargetisI,margetisII}.

\subsection{Scope and outline}
\label{subsec:objec}

Our objective in this work is to find an explicit approximation for the evolution
\begin{align*}
 e^{itH_N} \bpsi_0
\end{align*}
in the Fock space norm, where $\bpsi_0$ is the coherent state~\eqref{coherent}.
This would imply an approximation for the evolution
\begin{equation*}
e^{i t H_{N, N }} \psi_0
\end{equation*}
in $L^2(\R^{3N})$ as $N \to \infty$. To the best of our knowledge, no such approximation
is available in the mathematics or physics literature.
In particular, the tensor product type approximation \eqref{firstapprox}
for $\phi$ satisfying a Hartree equation, as in \cite{Rod-S}, is not
known to be such a Fock space approximation (nor do we expect it to be).

To accomplish our goal,
we propose to modify \eqref{firstapprox}  in two ways. One minor correction is the
multiplication by an oscillatory term. A second correction is a composition with a second-order ``Weyl operator''. Both corrections are inspired by the work of Wu~\cite{wuI,wuII}; see also \cite{MargetisI,margetisII}. However,
our set-up and derived equation is essentially different from these works.

We proceed to describe the second order correction.
Let $k(t, x, y)= k(t, y, x)$ be a function (or kernel) to be determined later,
with $k(0, x, y)=0$. The minimum regularity expected
 of $k$ is $k \in L^2(dx \, dy)$ for a.e. t.

We define the operator
\begin{align}\label{B-op}
B =\frac{1}{2} \int \left(k(t, x, y)a_x a_y - \overline k(t, x, y) a^*_x a^*_y\right)\ dx\, dy~.
\end{align}
Notice that $B$ is skew-Hermitian, i.e., $i B$ is self-adjoint. The operator $e^B$
could be defined by the spectral theorem; see \cite{R-N}.
However, we prefer the more direct approach of defining it first
on the dense subset of vectors with finitely many non-zero components,
where it can be defined by a convergent Taylor series if $\|k\|_{L^2(dx dy)}$ is sufficiently small. Indeed, $B$ restricted to the subspace
of vectors with all entries past the first $N$ identically zero has norm
$\le C N \|k\|_{L^2}$. Then $e^B$ is extended to $\F$ as a unitary operator.

Now we have described all ingredients needed to state our results and derivations.
The remainder of the paper is organized as follows. In section~\ref{sec:stat-proof} we state
our main result and outline its proof. In section~\ref{sec:hartree} we study implications
of the Hartree equation satisfied by the one-particle wave function $\phi(t,x)$.
In section~\ref{sec:algebra} we develop bookkeeping tools of Lie algebra for computing requisite operators containing $B$.
In section~\ref{sec:k-eq} we study the evolution equation for a matrix $K$ that involves the kernel $k$.
In section~\ref{sec:solns} we develop an argument for the existence of solution to the equation
for the kernel $k$. In section~\ref{sec:error-I} we find conditions under which terms involved
in the error term $e^BVe^{-B}$ are bounded. In section~\ref{sec:error-II} we study similarly the
error term $e^B[A,V]e^{-B}$. In section~\ref{sec:trace} we show that we can control traces needed in derivations.

\section{Statement of main result and outline of proof}
\label{sec:stat-proof}

 In this section we state our strategy for general potentials satisfying
 certain properties. Later in the paper we show that
 all assumptions of the related theorem are satisfied locally in time for $v(x)= \chi(x)\frac{\epsilon}{|x|}$, $\epsilon$: sufficiently small,
 and $\chi \in C_0^{\infty}$: even.

\begin{theorem} \label{main}
Suppose that $v$ is an even
potential.
Let $\phi$ be a smooth solution of the Hartree equation
\begin{align}
i \frac{\partial \phi}{\partial t} + \Delta \phi \label{H-F-phi}
+ (v*|\phi|^2)\phi =0
\end{align}
with initial conditions $\phi_0$, and assume the three conditions listed below:
\begin{enumerate}
\item \label{c1}
Assume that we have $k(t, x, y) \in L^2(dx dy)$ for a.e. $t$, where $k$ is symmetric,
and solves
\begin{align} \label{newnlsshort}
(i u_t + u g^T + g u - (1+p)m)  = (i p_t + [g, p] + u \overline m)(1+p)^{-1} u~,
\end{align}
where all products in \eqref{newnlsshort} are interpreted as spatial compositions of kernels, ``$1$" is the identity operator, and
\begin{align}
&u(t, x, y):= \sh:=k + \frac{1}{3!} k \overline k k + \ldots~, \label{defs}\\
&\delta(x-y)+p(t, x, y):= \ch:=\delta(x-y) + \frac{1}{2!} k \overline k + \ldots\notag~,\\
&g(t, x, y):= - \Delta_x \delta (x-y)
-v(x-y) \phi(t, x) \overline \phi(t, y) - (v * |\phi|^2 )(t, x) \notag \delta(x-y)~, \\
&m(t, x, y) := v(x-y)\overline \phi(t, x) \overline \phi(t, y)~. \notag
\end{align}
\item \label{c2}
Also, assume that the functions
\begin{align*}
f(t):= \|e^B[A, V]e^{-B} \Omega\|_{\F}
\end{align*}
and
\begin{align*}
g(t):=\|e^BV e^{-B}\Omega\|_{\F}
\end{align*}
are locally integrable ($V$ is defined  in \eqref{H}).
\item \label{c3}
Finally, assume that $\int d(t, x, x)\ dx$ is locally integrable in time, where
 \begin{align*}
d(t, x, y)=
&
\left(i \sh_t +\sh g^T+ g \sh\right) \shb \notag\\
-&\left(i \ch_t + [g, \ch]\right)\ch\notag\\
-&\sh \overline m \ch -\ch m \shb~.
\end{align*}
\end{enumerate}
Then, there exist  real functions $\chi_0$,  $\chi_1$ such that
\begin{align}
&\| e^{-\sqrt N A(t)} e^{-B(t)} e^{ -i \int_0^t ( N \chi_0(s) + \notag \chi_1(s))ds} \Omega - e^{itH_N} \bpsi_0 \|_{\F} \\
&\le
\frac{\int_0^t f(s) ds}{\sqrt N} +
\frac{\int_0^t g(s) ds}{ N}~.\label{l2est}
 \end{align}

Recall that we defined (see section~\ref{sec:intro})
\begin{align*}
&\bpsi_0= e^{-\sqrt N A(0)} \Omega  \, \,  \mbox{ an arbitrary coherent state (initial data)}~,\\
&A(t)=a( \overline \phi(t, \cdot)) - a^*( \phi(t, \cdot))~,\\
&B(t) =\frac{1}{2} \int \left(k(t, x, y)a_x a_y - \overline k(t, x, y) a^*_x a^*_y\right)\ dx\, dy~.
\end{align*}
\end{theorem}

A few remarks on Theorem~\ref{main} are in order.

\begin{remark}
Written explicitly, the left-hand side of \eqref{newnlsshort} equals
\begin{align*}
&i u_t + u g^T + g u - (1+p)m=
\left(i \frac{\partial}{\partial t}
- \Delta_x - \Delta_y \right) u(t, x, y)
\\
& - \phi(t, x) \int v(x-z) \overline \phi(t, z) u(t, z, y)\ dz
- \phi(t, y) \int u(t, x, z) v(z-y) \overline \phi(t,z)\ dz\\
&-(v* |\phi|^2)(t, x)u (t, x, y) -(v* |\phi|)^2(t, y)u (t, x, y) \\
&-v(x-y)\overline \phi(t, x) \overline \phi(t, y)\\
&- \overline \phi(t, y) \int(1+p)(t, x, z)v(z-y)\overline \phi(t, z)\ dz~.
\end{align*}
The main term in the right-hand side equals
\begin{align*}
&i p_t + [g, p] + u \overline m
=i \frac{\partial}{\partial t}  p(t, x, y) + \left( - \Delta_x + \Delta_y
\right) p(t, x, y)\\
& - \phi(t, x) \int v(x-z) \overline \phi(t, z) p(t, z, y)\ dz\\
&+ \phi(t, y) \int p(t, x, z) v(z-y) \overline \phi(t, z)\ dz\\
&-(v* |\phi|^2)(t, x)p (t, x, y) +(v* |\phi|)^2(t, y)p (t, x, y) \\
&+ \int u(t, x, z) v(z-y)  \phi(t, z)  \phi(t, x)\ dz~.\\
\end{align*}
\end{remark}
\begin{remark} The algebra, as well as the local analysis presented in this paper do not depend on the sign of $v$. However, the global in time analysis of our equations would require $v$ to be non-positive.
\end{remark}
\begin{remark} Our techniques would allow us to consider more general initial data of the form $\bpsi_0= e^{-\sqrt N A(0)} e^{-B(0)} \Omega$.
For convenience, we only consider the case of tensor products ($B(0)=0$)
in this paper.
\end{remark}

 \begin{proof}
Since $e^{\sqrt N A}$ and $e^B$ are unitary, the left-hand side of \eqref{l2est} equals
\begin{align*}
\|
e^{ i\int_0^t (N \chi_0(s) + \chi_1(s))ds}e^{B(t)} e^{\sqrt N A(t)} e^{itH_N} e^{-\sqrt N A(0)} \Omega -\Omega \|_{\F}~.
\end{align*}
Define
\begin{align*}
\Psi(t)=
e^{B(t)} e^{\sqrt N A(t)} e^{itH} e^{-\sqrt N A(0)} \Omega~.
\end{align*}
In Corollary \ref{maincor} of section~\ref{sec:k-eq} we show that our equations for $\phi$, $k$  insure that
\begin{align*}
\frac{1}{i} \frac{\partial}{\partial t} \Psi =  L \Psi~,
\end{align*}
where $L=\widetilde L - N \chi_0 - \chi_1$
for some
$\widetilde{L}$: Hermitian, i.e.
 $\widetilde{L} = \widetilde{L}^*$, where $\widetilde{L}$ commutes with functions of time, $\chi_0$, $\chi_1$
are real functions of time, and, most importantly (see corollary \ref{maincor} of section~\ref{sec:k-eq} and the remark following it),
\begin{align}
\|\widetilde{L} \Omega\|_{\F} \label{Lomega}
\le N^{-1/2} \|e^B [A, V] e^{-B} \Omega\|_{\F} + N^{-1} \|e^B V e^{-B} \Omega\|_{\F}~.
\end{align}
We apply energy estimates to
\begin{align*}
\left(\frac{1}{i} \frac{\partial}{\partial t}
 -\widetilde{L}\right)
(e^{i \int_0^t (N \chi_0(s) + \chi_1(s))
 ds}
\Psi - \Omega)=\widetilde{L} \Omega ~.
\end{align*}
Explicitly,
\begin{align*}
&\frac{\partial}{\partial t} \left(\|
 (e^{i \int_0^t (N \chi_0(s) + \chi_1)ds}\Psi - \Omega) \|^2_{\F}\right)\\
 &=
 2 \Re \left(\frac{\partial}{\partial t}(e^{i \int_0^t (N \chi_0(s) + \chi_1)ds}\Psi - \Omega),\, e^{i \int_0^t (N \chi_0(s) + \chi_1)ds}\Psi - \Omega \right)\\
 &=
 2 \Re \left(\left(\frac{\partial}{\partial t} - i \widetilde{L}\right)(e^{i \int_0^t (N \chi_0(s) + \chi_1)ds}\Psi - \Omega), \, e^{i \int_0^t (N \chi_0(s) + \chi_1)ds}\Psi - \Omega \right)\\
 &=
 2 \Re \left( i \widetilde{L}\Omega, \, e^{i \int_0^t (N \chi_0(s) + \chi_1)ds}\Psi - \Omega \right)\\
 &\le 2 \left( N^{-1/2} \|e^B [A, V] e^{-B} \Omega\|_{\F} + N^{-1} \|e^B V e^{-B} \Omega\|_{\F}\right)
 \|
 (e^{i \int_0^t (N \chi_0(s) + \chi_1)ds}\Psi - \Omega) \|_{\F}
 ~.
\end{align*}
Thus
\begin{align*}
&\frac{\partial}{\partial t} \|
 (e^{i \int_0^t (N \chi_0(s) + \chi_1)ds}\Psi - \Omega) \|\le
  N^{-1/2} \|e^B [A, V] e^{-B} \Omega\|_{\F} + N^{-1} \|e^B V e^{-B} \Omega\|_{\F}
  ~.
\end{align*}
and \eqref{l2est} holds. This concludes the proof.
\end{proof}
$\square$

\section{The Hartree equation}
\label{sec:hartree}

In this section we see how far we can go by using only the Hartree equation for the one-particle wave function $\phi$.

\begin{lemma} \label{Hartree-F}
The following commutation relations hold (where the $t$ dependence is suppressed, $A$ denotes $A(\phi)$ and $V$ is defined by formula \eqref{H}):
\begin{align}
&[A, V]= \notag
\int v(x-y) \left( \overline \phi (y) a^*_x a_x a_y + \phi (y) a^*_x a^*_y a_x \right)\ dx\, dy\\
&\big[A, [A, V] \big] \label{potential}\\
&=
\int v(x-y) \left(
\overline \phi (y) \overline \phi(x)  a_x a_y + \notag
\phi (y) \phi(x)  a^*_x a^*_y +
2 \overline \phi (y)  \phi(x)  a^*_x a_y \right)\ dx\, dy\\
&+ 2 \int \left(v*|\phi^2|\right)(x) a_x^* a_x\ dx \notag\\
&\Big[A, \big[A, [A, V] \big]\Big] \notag\\
&=6\int \left(v*|\phi^2|\right)(x) \left(\phi(x) a_x^* + \overline \phi(x) a_x \right)\  dx \notag\\
&\bigg[A,\Big[A, \big[A, [A, V] \big]\Big]\bigg] \notag\\
&=12\int \left(v*|\phi^2|\right)(x) |\phi(x)|^2\ dx~. \notag
\end{align}
\end{lemma}

\begin{proof} This is an elementary calculation and is left to the interested reader.
\end{proof}

Now, we consider
$\Psi_1(t)=
e^{\sqrt N A(t)} e^{itH} e^{- \sqrt N A(0)} \Omega $
for which we have the
basic calculation in the spirit of Hepp \cite{hepp}, Ginibre-Velo \cite{G-V}, and Rodnianski-Schlein~\cite{Rod-S};
see equation~(3.7) in \cite{Rod-S}.

\begin{proposition} If $\phi$ satisfies the Hartree equation
\begin{align*}
i \frac{\partial \phi}{\partial t} + \Delta \phi
+ (v*|\phi|^2)\phi =0
\end{align*}
while
\begin{align*}
\Psi_1(t)=
e^{\sqrt N A(t)} e^{itH} e^{- \sqrt N A(0)} \Omega~,
\end{align*}
then $\Psi_1(t)$ satisfies
\begin{align*}
&\frac{1}{i} \frac{\partial}{\partial t} \Psi_1(t)
=  \bigg(H_0
 +\frac{1}{2} [A, [A, V] ]\\
+& N^{-1/2} [A, V] + N^{-1}V
-\frac{N}{2} \int v(x-y) |\phi(t, x)|^2 |\phi(t, y)|^2 dx \, dy\bigg)
\Psi_1(t)~.
\end{align*}
\end{proposition}

\begin{proof}
Recall the formulas
\begin{align*}
\left(\frac{\partial}{\partial t} e^{ C(t)}\right)
\left(e^{-C(t)}\right) = \dot{C} + \frac{1}{2!}[C, \dot{C}]
+ \frac{1}{3!}\big[C, [C, \dot{C}] \big]+ \ldots
\end{align*}
and
\begin{align*}
 e^{C}  H e^{-C}= H + [C, H] +\frac{1}{2!}\big[C, [C, H]\big] + \ldots~.
\end{align*}
Applying these relations to $C=\sqrt N A$ we get

\begin{equation}
\frac{1}{i} \frac{\partial}{\partial t} \psi_1(t) \label{r-s}
=L_1 \psi_1~,
\end{equation}
where
\begin{align*}
&L_1 = \frac{1}{i} \left(\frac{\partial}{\partial t} e^{\sqrt N A(t)}\right)
e^{-\sqrt N A(t)} + e^{\sqrt N A(t)}  H e^{- \sqrt N A(t)}\\
=&\frac{1}{i}\left(  N^{1/2} \dot{A} + \frac{N}{2}[A, \dot{A}]\right)
+  H + N^{1/2} [A, H_0] \\
&+ N^{-1/2} [A, V]
+  \frac{N}{2} \big[A, [A , H_0] \big]    \\
&\frac{1}{2} \big[A, [A, V] \big]
+  \frac{N^{1/2}}{3!}\Big[A, \big[A, [A, V] \big]\Big] +
 \frac{N}{4!}\bigg[A, \Big[A, \big[A, [A, V] \big]\Big]\bigg]~.
\end{align*}

Eliminating the terms with a weight of $\sqrt N$, or setting
\begin{equation} \label{H-F}
\frac{1}{i}
\dot{A} + [A, H_0] + \frac{1}{3!}\Big[A, \big[A, [A, V] \big]\Big] =0~,
\end{equation}
is exactly equivalent to the Hartree equation \eqref{H-F-phi}.
By taking an additional bracket with $A$ in \eqref{H-F}, we have
\begin{align*}
\frac{1}{i}[A, \dot{A}] + \big[A, [A, H_0]\big] + \frac{1}{3!}\bigg[A, \Big[A, \big[A, [A, V] \big]\Big]\bigg] =0~,
\end{align*}
and thus simplify \eqref{r-s} to
\begin{align*}
&\frac{1}{i} \frac{\partial}{\partial t} \psi_1(t)
=  \Bigg( H_0
 +\frac{1}{2} [A, [A, V] ]\\
+& N^{-1/2} [A, V] + N^{-1}V
-N \frac{1}{4!}\bigg[A, \Big[A, \big[A, [A, V] \big]\Big]\bigg]
\Bigg) \psi_1~.
\end{align*}
This concludes the proof.
\end{proof}
$\square$

The first two terms on the right-hand side are the main ones.
The next two terms are $O\left(\frac{1}{\sqrt N}\right)$ and $O\left(\frac{1}{ N}\right)$.
The last term equals
 \begin{align*}
 -\frac{N}{2} \int v(x-y) |\phi(t, x)|^2 |\phi(t, y)|^2 dx \, dy
 := - N \chi_0~.
 \end{align*}

Notice that $\|L_1(\Omega)\|$ is not small because of the presence of $a^*_xa^*_y$ in $
\big[A, [A, V] \big]$. In order to eliminate these terms,
 we introduce $B$ (see \eqref{B-op}) and take
\begin{equation}
\psi=e^B \psi_1~. \notag
\end{equation}
Accordingly, we compute
\begin{equation}
\frac{1}{i} \frac{\partial}{\partial t} \psi = L \psi~, \notag
\end{equation}
where

\begin{align*}
&L=\frac{1}{i}\left(\frac{\partial}{\partial t} e^B\right)e^{-B}
+ e^B L_1
e^{-B}\\
&=L_Q + N^{-1/2} e^B [A, V] e^{-B} + N^{-1} e^B V e^{-B} - N \chi_0~,
\end{align*}
and
\begin{align}
L_Q=\frac{1}{i}\left(\frac{\partial}{\partial t} e^B\right)e^{-B}
+ e^B \left(  H_0
 +\frac{1}{2} \big[A, [A, V] \big] \right)e^{-B} \label{L_Q}
\end{align}
contains all quadratics in the operators $a$, $a^*$.

Equation \eqref{newnlsshort} for $k$ turns out to be
 equivalent to the requirement that
$L$ has no terms of the form $a^* a^*$ .
Terms of the form $a a^*$ will occur,  and will be converted to
$a^* a$ at the expense of $\chi_1$.

In other words, we require that $L_Q$
have no terms of the form $a^* a^*$. For a similar argument (but
for a different set-up), see Wu~\cite{wuII}.

\section{The Lie algebra of ``symplectic matrices"}
\label{sec:algebra}

In this section we describe the bookkeeping tools needed to compute
$L_Q$ of \eqref{L_Q} in closed form.
The results of this section are essentially standard, but they are included
here for the sake of completeness.

We start with the remark that
\begin{align}
&[a(f_1) + a^*(g_1), a(f_2) + a^*(g_2)]=\label{symp} \int f_1g_2 - f_2 g_1\\
&=-\left(
\begin{matrix}
f_1& g_1
\end{matrix}
\right)
 J
\left(
\begin{matrix}
f_2\\
g_2
\end{matrix}\notag
\right)
\end{align}
where
\begin{align*}
J=
\left(
\begin{matrix}
0& -\delta(x-y)\\
\delta(x-y)&0
\end{matrix}
\right).
\end{align*}
This observation explains why we have to invoke symplectic linear algebra.
We thus consider the infinite-dimensional Lie algebra $\mathit{sp}$
of ``matrices'' of the form

\begin{align*}
S(d, k, l)=
\left(
\begin{matrix}
d&k\\
l&-d^T
\end{matrix}
\right)
\end{align*}
for symmetric kernels $k=k(t, x, y)$ and $l=l(t, x, y)$, and arbitrary kernel $d(t, x, y)$.
(The dependence on $t$ will be suppressed when not needed.)
This situation is analogous to the Lie algebra of the finite-dimensional complex
symplectic group, with $x$, $y$ playing the role of $i$ and $j$.
We also consider the Lie algebra $\mathit{Quad}$ of quadratics of the form
\begin{align}
Q(d, k, l):=
&\frac{1}{2}\left( \begin{matrix} \label{metaplectic}
a_x& a_x^*
\end{matrix}
\right)
\left(
\begin{matrix}
d&k\\
l&-d^T
\end{matrix}
\right)
\left(
\begin{matrix}
-a_y^*\\
a_y
\end{matrix}
\right)\\
&=-
 \int d(x, y) \frac{
a_x a_y^* + a_y^* a_x}{2}\ dx\, dy \notag
+  \frac{1}{2}
\int k(x, y)a_x a_y\ dx\, dy\\
&- \frac{1}{2}
\int l(x, y)a_x^* a_y^*\ dx\, dy \notag
\end{align}
($k$, $l$ and $d$ as before). Furthermore, we agree to
identify operators which differ (formally) by a scalar operator.
Thus, $\int d(x, y) a_x a_y^*$ is considered
equivalent to $\int d(x, y)  a_y^* a_x$.
We recall the following result related to the metaplectic representation
(see, e.g. \cite{F}).

\begin{theorem} \label{algebra0}
Let $S=S(d, k, l)$, $Q=Q(d, k, l)$ related as above. Let $f$, $g$ be functions (or distributions).
 Denote
 \begin{align*}
 (a_x, a^*_x)
\left(
\begin{matrix}
f\\
g
\end{matrix}
\right) :=\int \left(f(x) a_x + g(x) a^*_x \right) dx~.
\end{align*}
We have the following commutation relation:
\begin{align}
[Q, (a_x, a^*_x) \label{meta1}
\left(
\begin{matrix}
f\\
g
\end{matrix}
\right)]=(a_x, a^*_x)S\left(
\begin{matrix}
f\\
g
\end{matrix}
\right)
\end{align}
where products are interpreted as compositions. We also have
\begin{align}
e^Q (a_x, a^*_x) \label{metaexp}
\left(
\begin{matrix}
f\\
g
\end{matrix}
\right)e^{-Q}=(a_x, a^*_x)e^S\left(
\begin{matrix}
f\\
g
\end{matrix}
\right)~,
\end{align}
provided that $e^Q$ makes sense as a unitary operator ($Q$: skew-Hermitian).
\end{theorem}

\begin{proof}
The commutation relation \eqref{meta1} can be easily checked directly,
but we point out that it follows from \eqref{symp}.
In fact, using \eqref{symp}, for any rank one quadratic we have
\begin{align*}
&[\left(a(f_1) + a^*(g_1)\right)\left(a(f_2)+ a^*(g_2)\right), a(f)+a^*(g)]\\
&=-\left( \begin{matrix}
a_x& a^*_x
\end{matrix}
\right)
\left(
\left(
\begin{matrix}
f_2\\
g_2
 \end{matrix} \right)
\left(
\begin{matrix}
f_1&
g_1
 \end{matrix}
 \right) +
\left( \begin{matrix}
f_1\\
g_1
 \end{matrix}
 \right)
\left( \begin{matrix}
f_2&
g_2
\end{matrix}
\right)
\right)
J
\left( \begin{matrix}
f\\
g
 \end{matrix}
 \right)~.
\end{align*}
Thus, for any $R$ we have
\begin{align*}
[
\left( \begin{matrix}
a_x& a^*_x
\end{matrix}
\right) R
\left( \begin{matrix}
a_y\\ a^*_y
\end{matrix}
\right)
, a(f) + a^*(g)]
=-
\left( \begin{matrix}
a_x& a^*_x
\end{matrix}
\right) \left( R + R^T \right) J
\left( \begin{matrix}
f\\
g
\end{matrix}
\right)~.
\end{align*}
Now specialize to $R=\frac{1}{2} S J$, $S \in \mathit{sp}$, and use
$S^T=JSJ$ to complete the proof.

The second part, equation \eqref{metaexp}, follows from the identity
\begin{align*}
 e^{Q}  C e^{-Q}= C + [Q, C] +\frac{1}{2!}\big[Q, [Q, C]\big] + \ldots~,
\end{align*}
or, in the language of adjoint representations,
$\mbox{Ad}(e^Q)(C)= e^{\mbox{ad}(Q)}(C)$, which is applied to
$C=a(f) + a^*(g)$.
\end{proof}
$\square$

A closely related result is provided by the following theorem.

\begin{theorem} \label{algebra}
\begin{enumerate}
\item The linear map $\I: \mathit{sp} \to \mathit{Quad}$
defined by

\begin{align*}
S(d, k, l) \to Q(d, k, l)
\end{align*}
is a Lie algebra isomorphism.

\item Moreover, if $S=S(t)$, $Q=Q(t)$ and $\I(S(t))=Q(t)$
is skew-Hermitian, so that $e^Q$ is well defined,
we have

\begin{align}
\I\left(\left(\frac{\partial}{\partial t} e^S\right) e^{-S}\right)
=
\left(\frac{\partial}{\partial t} e^Q\right) e^{-Q} \label{diff}~.
\end{align}

\item  Also, if $R \in \mathit{sp}$, we have
\begin{align}
\I\left( e^S R e^{-S}\right) =  e^Q \I( R) e^{-Q} \label{conj}~.
\end{align}

\end{enumerate}

\end{theorem}

\begin{remark}

In the finite-dimensional case, this is (closely related to) the
 ``infinitesimal metaplectic representation''; see
p. 186 in \cite{F} . In the infinite dimensional case, we must be careful, as some of our operators are not of trace class. For instance,
$\int a_x  a_x^*$ does not make sense.

\end{remark}

\begin{proof} First, we point out that \eqref{metaexp}
implies \eqref{conj}, at least in the case where $R$ is the ``rank one" matrix
\begin{align*}
R=
\left(
\begin{matrix}
f\\
g
\end{matrix}
\right)
\left(
\begin{matrix}
h & i
\end{matrix}
\right)~.
\end{align*}

Notice that \eqref{metaexp} can also be written as
\begin{align*}
e^Q
\left(
\begin{matrix}
f & g
\end{matrix}
\right)
\left(
\begin{matrix}
a_x\\
a^*_x
\end{matrix}
\right)e^{-Q}=
\left(
\begin{matrix}
f & g
\end{matrix}
\right)
e^{S^T}\left(
\begin{matrix}
a_x\\
a^*_x
\end{matrix}
\right)~.
\end{align*}

In conclusion, we find
\begin{align*}
&e^Q
\left( \begin{matrix}
a_x& a_x^*
\end{matrix}
\right)
R
\left(
\begin{matrix}
-a^*_y\\
a_y
\end{matrix}
\right)e^{-Q}\\
&=e^Q
\left( \begin{matrix}
a_x& a_x^*
\end{matrix}
\right)
\left(
\begin{matrix}
f\\
g
\end{matrix}
\right)
\left(
\begin{matrix}
h & i
\end{matrix}
\right) J
\left(
\begin{matrix}
a_y\\
a^*_y
\end{matrix}
\right)e^{-Q}\\
&=e^Q
\left( \begin{matrix}
a_x& a_x^*
\end{matrix}
\right)
\left(
\begin{matrix}
f\\
g
\end{matrix}
\right) e^{-Q} e^{Q}
\left(
\begin{matrix}
h & i
\end{matrix}
\right) J
\left(
\begin{matrix}
a_y\\
a^*_y
\end{matrix}
\right)e^{-Q}\\
&=
\left( \begin{matrix}
a_x& a_x^*
\end{matrix}
\right) e^S
\left(
\begin{matrix}
f\\
g
\end{matrix}
\right)
\left(
\begin{matrix}
h & i
\end{matrix}
\right) J e^{J S J}
\left(
\begin{matrix}
a_y\\
a^*_y
\end{matrix}
\right)\\
&=
\left( \begin{matrix}
a_x& a_x^*
\end{matrix}
\right) e^S
R e^{-S}
\left(
\begin{matrix}
-a^*_y\\
a_y
\end{matrix}
\right)
\end{align*}
since
 $S^T=J S J$ if
 $S \in \mathit{sp}$, and
 $J e^{JSJ}= e^{-S}J$.

We now give a direct proof that \eqref{metaplectic} preserves Lie brackets.
Denote the quadratic building blocks by
$Q_{xy}= a_x a_y$, $Q^*_{xy}= a^*_x a^*_y$, $N_{xy}=\frac{1}{2} \left(a_xa^*_y + a^*_y a_x \right)$.
One can verify the following commutation relations, which will be also needed below:
\begin{align}
\big[Q_{xy}, Q^*_{zw} \big]&= \delta(x-z) N_{yw} \label{QQ^*}
+ \delta(x-w) N_{yz} + \delta(y-z) N_{xw} + \delta(y-w) N_{xz}~,\\
\big[Q_{xy}, N_{zw} \big]&= \delta(x-w) Q_{yz} + \delta(y-w) Q_{xz}
\label{QN}~,\\
\big[N_{xy}, Q^*_{zw} \big]&= \delta(x-z) Q^*_{yw} + \delta(x-w) Q_{yz}
\label{NQ^*}~,\\
\big[N_{xy}, N_{zw} \big]&= \delta(x-w) N_{zy} -\delta(y-z)N_{xw}~.\label{NN}
\end{align}

Using \eqref{QQ^*} we compute
\begin{align*}
\Big[\frac{1}{2}\int k(x, y)a_xa_y dx dy, - \frac{1}{2} \int l(x, y) a^*_x a^*_y dx dy \Big] = - \int(kl)(x, y) N_{xy}\ dx\, dy~,
\end{align*}
which corresponds to the relation
\begin{align*}
&\Bigg[
\left(
\begin{matrix}
0&k\\
0&0
\end{matrix}
\right),
\left(
\begin{matrix}
0&0\\
l&0
\end{matrix}
\right) \Bigg]\\
&=
\left(
\begin{matrix}
kl&0\\
0&-lk
\end{matrix}
\right).
\end{align*}
The other three cases are similar.

To prove \eqref{diff}, expand both the left-hand side and the right-hand side as

\begin{align*}
&\I\left(\left(\frac{\partial}{\partial t} e^S\right) e^{-S}\right)\\
&=\I\left( \dot{S} + \frac{1}{2} [S, \dot{S}] + \cdots \right)\\
&= \dot{Q} + \frac{1}{2} [Q, \dot{Q}] + \cdots \\
&=\left(\frac{\partial}{\partial t} e^Q\right) e^{-Q}~.
\end{align*}

The proof of \eqref{conj} is along the same lines.
\end{proof}
$\square$

\remark{Note on rigor}: All the Lie algebra results
that we have used are standard in the finite-dimensional case. In our applications, $S$ will be $K$ where $K$ is a matrix of the form
\eqref{B}, see below,
and
$Q$ will be $B=\I(K)$.
The unbounded operator $B$ is skew-Hermitian and $e^{B} \psi$ is defined by a convergent Taylor series if $\psi \in \F$ has only finitely many non-zero components, provided $\|k(t, \cdot, \cdot)\|_{ L^2(dx \, dy)}$ is small .
We then extend $e^B$ to all $\F$ as a unitary operator.
The norm $\|k(t, \cdot, \cdot)\|_{ L^2(dx \, dy)}$ iterates under compositions; thus, the kernel $e^K$ is well defined
by its convergent Taylor expansion. In the expression
\begin{align}
e^B P e^{-B} = P + [B, P] + \ldots \label{rigor}
\end{align}
for $P$, a first- or second-order polynomial in $a$, $a^*$, we point out that the right-hand side stays a polynomial of the same degree, and converges
when applied to a Fock space vector with finitely many non-zero components.
For our application, we need to know \eqref{rigor} is true when applied to $\Omega$.
The same comment applies to the series
\begin{align*}
\left(\frac{\partial}{\partial t} e^B \right) e^{-B}=\dot{B} +
\frac{1}{2}[B, \dot{B}]+ \ldots~.
\end{align*}

\section{Equation for kernel $k$}
\label{sec:k-eq}

Now apply the isomorphism of the previous section to the operator
\begin{align*}
B=\I(K)
\end{align*}
for
\begin{align}
K= \left(
\begin{matrix}
0&k(t, x, y)\\ \label{B}
\overline k(t, x, y)&0
\end{matrix}
\right)~.
\end{align}
This agrees to the letter with the isomorphism \eqref{metaplectic}.
The next two isomorphisms,
 \eqref{H_0} and \eqref{aav},
 require special treatment because  $a a^*$ terms mirroring the
$a^*a$ terms are missing in \eqref{H}, \eqref{potential}. However,
the discrepancy only happens on the diagonal.
Once the relevant terms are commuted with $B$, they fit the pattern exactly.
It isn't quite true that
\begin{align}
H_0=
&\I \left(\left(
\begin{matrix}
-(\Delta \delta)(x-y)&0\\ \notag
0&(\Delta \delta)(x-y)\\
\end{matrix}
\right)\right)\\
=&\I \left(\left(
\begin{matrix} \label{H_0}
-\Delta &0\\
0&\Delta \\
\end{matrix}
\right)\right)
\end{align}
since, strictly speaking,
\begin{align*}
&\I \left(\left(
\begin{matrix}
-(\Delta \delta)(x-y)&0\\ \notag
0&(\Delta \delta)(x-y)\\
\end{matrix}
\right)\right)
= \int \frac{a^*_x \Delta a_x + a_x \Delta a^*_x}{2} dx
\end{align*}
is undefined.
However, one can compute directly that
 $[ \Delta_x a_x, a^*_y]= (\Delta \delta) (x-y) ~.$
Using that, we compute
\begin{align*}
[B, H_0]
= \frac{1}{2} \int \left( (\Delta_x + \Delta_y)k(x, y) a_x a_y
+ (\Delta_x + \Delta_y)\overline k(x, y) a_x^* a_y^*\right)\ dx\, dy~.
\end{align*}

This commutator is in agreement with \eqref{B}, \eqref{H_0}, and the result can be represented in accordance with \eqref{metaplectic}, namely

\begin{align*}
[B, H_0]=
\I\left(\bigg[
\left(
\begin{matrix}
0& k \\
\overline k  &0\\
\end{matrix}\right),
\left(
\begin{matrix}
-(\Delta \delta)(x-y)&0\\ \notag
0&(\Delta \delta)(x-y)\\
\end{matrix}
\right)\bigg]
\right)~.
\end{align*}
We also have
\begin{align*}
&e^B H_0 e^{-B} - H_0\\
&=\I \left(e^K\left(
\begin{matrix}
-(\Delta \delta)(x-y)&0\\ \notag
0&(\Delta \delta)(x-y)\\
\end{matrix}
\right) e^{-K} -
\left(
\begin{matrix}
-(\Delta \delta)(x-y)&0\\ \notag
0&(\Delta \delta)(x-y)\\
\end{matrix}
\right) \right)
\end{align*}
since $e^B H_0 e^{-B} - H_0= [B, H_0] + \frac{1}{2}\big[B, [B, H_0] \big] +\cdots$
The same comment applies to the diagonal part of
\begin{align}
&\frac{1}{2} \big[A, [A, V] \big]= \notag\\
&\I\left(
\begin{matrix}
-v_{12} \overline \phi_1  \phi_2 - \left(v * |\phi|^2 \right) \delta _{12}
&v_{12}\overline \phi_1 \overline \phi_2\\
-v_{12}\phi_1\phi_2&v_{12}  \phi_1 \overline \phi_2 +
\left(v * |\phi|^2 \right) \delta_{12} \label{aav}
\\
\end{matrix}
\right),
\end{align}
where
 $v_{12}  \phi_1  \phi_2 $ is an abbreviation for the product
$v(x-y) \phi(x) \phi(y)$, etc. Formula \eqref{aav} isn't quite true either,
but becomes true after commuting with $B$.

To apply our isomorphism, we quarantine the ``bad" terms
in \eqref{H_0} and the diagonal part of \eqref{aav}.
Define
\begin{align*}
G=\left(
\begin{matrix}
g & 0\\
0 &- g^T
\end{matrix}
\right) \qquad
\mbox{
and} \qquad
M=\left(
\begin{matrix}
0& m\\
- \overline m &0
\end{matrix}
\right)
\end{align*}
where
\begin{align*}
&g= - \Delta \delta_{12}
-v_{12} \overline \phi_1  \phi_2 - (v * |\phi|^2 ) \delta_{12}~, \\
&m = v_{12}\overline \phi_1 \overline \phi_2~,
\end{align*}
and split
\begin{align*}
H_0
 +\frac{1}{2} \big[A, [A, V] \big]= H_G + \I(M)
 \end{align*}
 where
\begin{align}
H_G=H_0 +
&\int v(x-y) \notag
\overline \phi (y)  \phi(x)  a^*_x a_y \ dx\, dy \\
+&  \int \left(v*|\phi^2|\right)(x) a_x^* a_x\ dx~.\label{H_G}
\end{align}

By the above discussion we have
\begin{align*}
&[B, H_G]=\I([K, G]) \qquad \mbox{and}\\
&[e^B, H_G]e^{-B}=\I([e^K, G]e^{-K})~.
\end{align*}
Write
\begin{align}\notag
L_Q=&\frac{1}{i}\left(\frac{\partial}{\partial t} e^B\right) e^{-B}
\notag\\
+& e^B \left(  H_0
 +\frac{1}{2} \big[A, [A, V] \big] \right)e^{-B}
\notag\\
=&\frac{1}{i}\left(\frac{\partial}{\partial t} e^B\right) e^{-B}
\notag\\
+& H_G +[ e^B, H_G]e^{-B} + e^B \I(M) e^{-B}
\notag\\
=&  H_G+\I\left(
\left(\frac{1}{i} \frac{\partial}{\partial t}
e^{K}\right) e^{-K}
+[e^{K}, G]e^{-K} + e^K M e^{-K} \right)\notag\\
&=H_G + \I(\mathcal{M}_1+\mathcal{M}_2+\mathcal{M}_3)
 \label{finalform}~.
\end{align}

Notice that
if $K$ is given by \eqref{B}, then
\begin{align*}
e^{K}=
\left(
\begin{matrix}
\ch & \sh\\
\shb & \chb\\
\end{matrix}
\right),
\end{align*}
where
\begin{align}\label{ch}
\ch  = I + \frac{1}{2} k \overline{k} + \frac{1}{4!}
k \overline{k} k \overline{k} + \ldots~,
\end{align}
and similarly for $\sh$. Products are interpreted, of course, as compositions
of operators.

We compute
\begin{align*}
& \mathcal{M}_1=
\frac{1}{i}
 \left(
\begin{matrix}
\ch_t &\sh_t\\
\shb_t & \chb_t
\end{matrix}
\right)
 \left(
\begin{matrix}
\ch &-\sh\\
-\shb & \chb
\end{matrix}
\right)\\
=&\frac{1}{i}
 \left(
\begin{matrix}
\ch_t \ch - \sh_t \shb &-\ch_t \sh + \sh_t \chb \\
* & *
\end{matrix}
\right)
\end{align*}
\begin{align*}
[e^{K}, G]
=
\left(
\begin{matrix}
[\ch, g] & \, \, \, -\sh g^T -g \sh\\
* & *
\end{matrix}
\right)
\end{align*}

and

\begin{align*}
&\mathcal{M}_2=
[e^{K}, G]e^{-K}=\\
&
\left(
\begin{matrix}
[\chh, g] \chh + (\shh g^T+ g \shh) \shhb
&\, \, \,  -[\chh, g] \shh -( \shh g^T +  g \shh) \chhb\\
* & *
\end{matrix}
\right)~,
\end{align*}
where $\shh$ is an abbreviation for $\sh$, etc,
and
\begin{align*}
&\mathcal{M}_3=
e^K M e^{-K}=
\left(
\begin{matrix}
-\shh \overline m \chh -\chh m \shhb& \, \, \, \shh \overline m \shh + \chh m \chhb\\
*&*
\end{matrix}
\right)~.
\end{align*}
Now define
\begin{align*}
\mathcal{M} =\mathcal{M}_1+\mathcal{M}_2+\mathcal{M}_3~.
\end{align*}
We have proved the following theorem.

\begin{theorem} Recall
  the
isomorphism \eqref{metaplectic} of Theorem \ref{algebra}.

\begin{enumerate}
\item If $L_Q$ is given by \eqref{L_Q}, then
\begin{align}
L_Q=& H_0 +
\int v(x-y)
\overline \phi (y)  \phi(x)  a^*_x a_y \ dx\, dy \label{G}\\
+&  \int \left(v*|\phi^2|\right)(x) a_x^* a_x\ dx
+\I\left( \mathcal{M} \right)~.\notag
\end{align}
\item
The coefficient of $a_x a_y$ in
$ \I \left(\mathcal{M} \right)$ is
 $-\mathcal{M}_{12}$ or
\begin{align*}
&(i \sh_t + \sh g^T + g \sh)\chb -(i \ch_t -[\ch, g])\sh\\
&- \sh \overline m \sh - \ch m \chb~.
\end{align*}
\item
The coefficient of $a^*_x a^*_y $
equals minus the complex conjugate
of the coefficient  of $a_x a_y$.
\item
The coefficient of $- \frac{a_x a^*_y + a^*_y a_x}{2}$ is
$\mathcal{M}_{11} $, or
\begin{align}
d(t, x, y)=
&
\left(i \sh_t +\sh g^T+ g \sh\right) \shb \notag\\
-&\left(i \ch_t + [g, \ch]\right)\ch\notag\\
-&\sh \overline m \ch -\ch m \shb~. \label{trace}
\end{align}
\end{enumerate}

\end{theorem}

\begin{corollary} \label{maincor}
If $\phi$ and $k$ satisfy  \eqref{H-F-phi} and \eqref{newnlsshort} of
theorem \eqref{main}, then the coefficients of $a_x a_y$ and
$a^*_x a^*_y$ drop out and $L_Q$ becomes
\begin{align*}
L_Q=&
H_0 +
\int v(x-y)
\overline \phi (t, y)  \phi(t, x)  a^*_x a_y \ dx\, dy
+  \int \left(v*|\phi^2|\right)(x) a_x^* a_x\ dx\\
-& \int d(t, x, y) \frac{ a_x a^*_y + a^*_y a_x }{2}\ dx\, dy~,
\end{align*}
where $d$ is given by \eqref{trace} and the full operator reads
\begin{align*}
L=&
H_0 +
\int v(x-y)
\overline \phi (y)  \phi(t, x)  a^*_x a_y  dx dy
+  \int \left(v*|\phi^2|\right)(x) a_x^* a_x dx\\
-& \int d(t, x, y)  a^*_y a_x dx +
 N^{-1/2} e^B [A, V] e^{-B} + N^{-1} e^B V e^{-B} - N \chi_0 - \chi_1\\
 &:=\widetilde{L} - N \chi_0 - \chi_1~,
\end{align*}
and
\begin{align*}
 \chi_0=\frac{1}{2} \int v(x-y) |\phi(t, x)|^2 |\phi(t, y)|^2 dx \, dy~,
 \end{align*}
\begin{align*}
\chi_1(t)= -\frac{1}{2}\int d(t, x, x) dx~.
\end{align*}
\end{corollary}
\begin{remark}
Notice that
\begin{align*}
\widetilde{L} \Omega
= \left(
N^{-1/2} e^B [A, V] e^{-B} + N^{-1} e^B V e^{-B} \right) \Omega~,
\end{align*}
and therefore we can derive the bound
\begin{align*}
\|\widetilde{L} \Omega\|
\le N^{-1/2} \|e^B [A, V] e^{-B} \Omega\| + N^{-1} \|e^B V e^{-B} \Omega\|~.
\end{align*}
Also, $L$ is (formally) self-adjoint by construction.
The kernel $d(t, x, y)$, being the sum of the (1,1) entry of the self-adjoint matrices $
\left( \frac{1}{i} \frac{\partial}{\partial t} e^K \right) e^{-K}
$,
$[e^K, G]e^{-K} = e^K G e^{-K} -G$ and
the visibly self-adjoint term $ -\sh \overline m \ch -\ch m \shb $, is self-adjoint; thus, it has a real trace.
Hence, $\widetilde{L}$ is also self-adjoint.
\end{remark}

In the remainder of this paper, we check that the hypotheses of
our main theorem are satisfied, locally in time, for the potential $v(x) =
\chi(x) \frac{\epsilon}{|x|}$.

\section{Solutions to equation \eqref{newnlsshort}}
\label{sec:solns}

\begin{theorem} \label{localex}
Let $\epsilon_0$ be sufficiently small and assume that $v(x)= \frac{\epsilon_0}{|x|}$, or
$v(x)= \chi(x) \frac{\epsilon_0}{|x|}$ for $\chi \in C_0^{\infty}(\Bbb R^3)$ . Assume that $\phi$ is a smooth solution to
the Hartree equation \eqref{H-F}, $\|\phi\|_{L^2(dx)}=1$. Then there exists
$k \in L^{\infty}([0, 1]) L^2 (dx dy)$ solving \eqref{newnlsshort} with initial conditions $k(0, x, y)=0$  for
$0 \le t \le 1$.
The solution $k$ satisfies the following additional properties.
\begin{enumerate}
\item
\begin{align*}
\| \left( i
\frac{\partial}{\partial t} -\Delta_x -\Delta_y \right) k \|_{ L^{\infty}[0, 1] L^2(dx dy)} \le C~.
\end{align*}
\item
\begin{align*}
\| \left( i
\frac{\partial}{\partial t} -\Delta_x -\Delta_y \right) \sh \|_{ L^{\infty}[0, 1] L^2(dx dy)} \le C~.
\end{align*}
\item
\begin{align*}
\| \left( i
\frac{\partial}{\partial t} -\Delta_x +\Delta_y \right) p \|_{ L^{\infty}[0, 1] L^2(dx dy)} \le C~.
\end{align*}
\item The kernel $k$ agrees on $[0, 1]$ with a kernel $\widetilde{k}$ for which
\begin{align*}
 \|\widetilde{k}\|_{X^{\frac{1}{2}, \frac{1}{2}+}} \le C ~;
\end{align*}
see
\eqref{Xsd} for the definition of the space $ X^{s, \delta}$ and, of course, $\frac{1}{2}+$ denotes a fixed number slightly bigger than $\frac{1}{2}$.
\end{enumerate}
\end{theorem}

\begin{proof}
We first establish some notation.
Let $S$ denote the Schr\"odinger operator
\begin{equation*}
S=
i \frac{\partial}{\partial t}
- \Delta_x - \Delta_y
\end{equation*}
and let $ T $ be the transport operator
\begin{equation*}
T=
i \frac{\partial}{\partial t}
- \Delta_x + \Delta_y~.
\end{equation*}
Let $\epsilon : L^2(dx dy) \to L^2(dx dy)$ denote schematically any linear operator of operator norm $\le C \epsilon_0$, where $C$ is a ``universal constant". In practice, $\epsilon$ will be (composition with) a kernel of the type
$\phi(t, x) \phi(t, y) v(x-y)$, or multiplication by $v * |\phi|^2$.
Also, recall the inhomogeneous term
\begin{align*}
m(t, x, y)=v(x-y) \overline \phi(t, x)\overline \phi(t, y)~.
\end{align*}
Then, equation \eqref{newnlsshort}, written explicitly, becomes
\begin{equation}
S k  =m + S(k-u) + \epsilon(u)+ \epsilon(p) +(T p + \epsilon(p) + \epsilon(u))(1+p)^{-1}u~. \label{eqp}
\end{equation}
Note that $\ch^2 - \sh \shb = 1$; thus, $1+p = \ch \ge 1$ as an operator and  $(1+p)^{-1}$ is bounded  from $L^2$ to
$L^2$.
We plan to iterate in the norm $N(k)=\|k\|_{L^{\infty}[0, 1]L^2(dx dy)}
+\|S k\|_{L^{\infty}[0, 1]L^2(dx dy)}$.
Notice that $\|m\|_{L^2(dx dy)} \le C \epsilon_0$.

Now solve
\begin{align*}
Sk_0 = m
\end{align*}
with initial conditions $k_0(0, \cdot, \cdot)=0$, where $N(k_0) \le C \epsilon_0$.
Define $u_0$, $p_0$ corresponding to $k_0$.

For the next iterate, solve
 \begin{equation*}
S k_1  = m + S(k_0-u_0) + \epsilon(u_0)+ \epsilon(p_0) +(T p_0 + \epsilon(p_0) + \epsilon(u_0))(1+p_0)^{-1}u_0~;
\end{equation*}
the non-linear terms satisfy
\begin{align*}
&\|S (u_0-k_0)\|_{ L^{\infty}[0, 1]L^2(dx dy)}=\\
&\| \frac{1}{3!}\left( (Sk_0)\overline k_0 k_0 - k_0 \overline{(Sk_0)} k_0
+ k_0\overline k_0 Sk_0 \right) + \cdots\|_{L^{\infty}[0, 1]L^2(dx dy)}\\
&= O(N(k_0)^3)~.
\end{align*}
Also, recalling that
$p_0 =\mbox{ch}(k_0) -1$, we have
\begin{align*}
\|T( p_0)\|_{L^{\infty}[0, 1]L^2(dx dy)}&=\| \frac{1}{2} \left((S k_0) \overline k_0 - k_0 \overline{ (S k_0)}\right) + \cdots \|_{L^{\infty}[0, 1]L^2(dx dy)}\\
&= O\left(N(k_0)^2\right)~.
\end{align*}
Thus, $N(k_1) \le C \epsilon_0 + C \epsilon_0^2$.
Continuing this way, we obtain a fixed point solution in this space which satisfies
the first three requirements of theorem \ref{localex}.

In fact, we can apply the same argument to $\left(\frac{\partial}{\partial t}\right)^N D^a k$, since $\left(\frac{\partial}{\partial t}\right)^N D^a m \in L^{\infty}[0, 1]L^2 (dx \, dy)$ for $0 \le a < \frac{1}{2}$.
However, we cannot repeat the argument for
$D^{1/2}k$.

We would like to have $\|S D^{1/2}k\|_{L^{\infty}[0, 1]L^2 (dx \, dy)}$  finite.
Unfortunately, this misses ``logarithmically" because of the singularity of $v$.

Fortunately, we can use the well-known $X^{s, \delta}$ spaces (see \cite{B,KPV,K-MM})
to show that $\||S|^s D^{1/2}u\|_{L^2(dt)L^2 (dx \, dy)}$
is finite locally in time for (all) $1>s> \frac{1}{2}$. This assertion will be sufficient for our purposes.
Recall the definition of $X^{s, \delta}$:
\begin{align}
\|
|\xi|^{s}\big(|\tau -|\xi|^2| +1\big)^{\delta}\widehat{ u}
\|_{L^{2}(d\tau d\xi)} :=\Vert u\Vert_{X^{s ,\delta}}~. \label{Xsd}
\end{align}
Going back to \eqref{eqp}, we write
\begin{align*}
S(k)=m + F
\end{align*}
where
we define the expression
\begin{align*}
F(k):=S(k-u)-\epsilon(u)+pm +\left(T(p)+\epsilon(p)+u\overline{m}\right)(1+p)^{-1}u\ .
\end{align*}
The idea is to
localize in time on the right-hand side:
\begin{align*}
S(\widetilde{k})=\chi(t) \left(m + F\right)~,
\end{align*}
where $\chi \in C_0^{\infty}(\Bbb R)$, $\chi=1$ on $[0, 1]$.
Then, $\widetilde{k}=k$ on $[0, 1]$.

As we already pointed out, we can  estimate
$\|S \left(\frac{\partial}{\partial t}\right)^N D^a k\|_{L^2 [0, 1]L^2 (dx \, dy)} \le C$ for $0 \le a < \frac{1}{2}$.
We can further localize $\widetilde{k}$ in time to insure that these relations hold globally
in time.
By using the triangle inequality $|\tau - |\xi|^2| + |\tau| \ge |\xi|^2$,
we immediately conclude that
\begin{align*}
\|
|\xi|^{\frac{3}{2}-}\big(|\tau -|\xi|^2| +1\big)^{\frac{1}{2}+}\widehat{ k_{\chi}}
\|_{L^{2}(d\tau d\xi)} \le C~.
\end{align*}

\end{proof}
$\square$

\section{Error term $e^B V e^{-B}$\label{errors1}}
\label{sec:error-I}

The goal of this section is to list explicitly
all terms in $e^B V e^{-B}$  and to find conditions under which these terms are bounded.
Recall that $V$ is defined by $V=
\int v(x_0 - y_0) Q^*_{x_0y_0}Q_{x_0y_0}\ d x_0\, d y_0$.
For simplicity,  $\shbb$ denotes either $\sh$ or $\shb$, and $\chbb$ denotes either $\ch$ or $\chb$.

Let $x_0 \neq y_0$; we obtain
\begin{align*}
e^B Q^*_{x_0y_0} Q_{x_0y_0} e^{-B}=e^B Q^*_{x_0y_0} e^{-B} e^B Q_{x_0y_0} e^{-B}~.
\end{align*}
According to the isomorphism \eqref{metaplectic}, we have
\begin{align*}
Q^*_{x_0 y_0} = \I
\left(
\begin{matrix}
0& \,0\\
-2 \delta (x- x_0)\delta (y- y_0)&\,0
\end{matrix}
\right)
\end{align*}
where the operator
\begin{align*}
&e^B Q^*_{x_0y_0} e^{-B}\\
&=\I \left(
\left(
\begin{matrix}
\ch & \sh\\
\shb& \chb
\end{matrix}
\right)
\left(
\begin{matrix}
0& \,0\\
-2 \delta (x- x_0)\delta (y- y_0)&\,0
\end{matrix}
\right)
\left(
\begin{matrix}
\ch & -\sh\\
-\shb& \chb
\end{matrix}
\right) \right)
\end{align*}
is a linear combination of the terms
\begin{align}
&\int \chbb (x, x_0) \chbb (y_0, y)\label{Q^*_{xy}}
 Q^*_{x y} dx \, dy~,\\
&\int \shbb (x, x_0) \chbb (y_0, y)\notag
 N_{x y} dx \, dy~,  \\
&\int \shbb (x, x_0) \shbb (y_0, y)\notag
 Q_{x y} dx \, dy~.\\ \notag
\end{align}

A similar calculation shows that $e^B Q_{x_0y_0} e^{-B}$
is a linear combination of
\begin{align}
&\int \chbb (x, x_0) \chbb (y_0, y) \label{Q_{xy}}
 Q_{x y} dx \, dy~,\\
&\int \shbb (x, x_0) \chbb (y_0, y)\notag
 N_{x y} dx \, dy~,\\
&\int \shbb (x, x_0) \shbb (y_0, y)\notag
 Q^*_{x y} dx \, dy~.\\ \notag
\end{align}

 Thus, $e^B Q^*_{x_0y_0}Q_{x_0y_0}e^{-B}$ is a linear combination of the
 nine possible terms obtained by combining the above.

   Now we list all terms in $e^B V e^{-B} \Omega$.
   Terms in  $e^B V e^{-B} $ ending in $Q_{xy}$ are automatically discarded because they contribute nothing when applied to $\Omega$.
 The remaining six terms are listed below.
 \begin{align}
\int &\chbb (x_1, x_0) \chbb (y_0, y_1) \notag
\shbb (x_2, x_0) \chbb (y_0, y_2)\\
& v(x_0 - y_0) Q^*_{x_1 y_1} N_{x_2 y_2} \Omega
 dx_1 \, dy_1 \, dx_2 \, dy_2\, dx_0 \,dy_0~, \label{one}
 \end{align}
 \begin{align}
 \int &\chbb (x_1, x_0) \chbb (y_0, y_1) \notag
\shbb (x_2, x_0) \shbb (y_0, y_2)\\
& v(x_0 - y_0) Q^*_{x_1 y_1} Q^*_{x_2 y_2} \Omega
 dx_1 \, dy_1 \, dx_2 \, dy_2\, dx_0 \,dy_0~, \label{two}
 \end{align}
  \begin{align}
\int &\shbb (x_1, x_0) \chbb (y_0, y_1) \notag
\shbb (x_2, x_0) \chbb (y_0, y_2)\\
& v(x_0 - y_0) N_{x_1 y_1} N_{x_2 y_2} \Omega
 dx_1 \, dy_1 \, dx_2 \, dy_2\, dx_0 \, dy_0~, \label{three}
 \end{align}
 \begin{align}
 \int &\shbb (x_1, x_0) \chbb (y_0, y_1) \label{four}
\shbb (x_2, x_0) \shbb (y_0, y_2)\\
& v(x_0 - y_0) N_{x_1 y_1} Q^*_{x_2 y_2} \Omega \notag
 dx_1 \, dy_1 \, dx_2 \, dy_2 \, dx_0 \,dy_0~,
 \end{align}
 \begin{align}
\int &\shbb (x_1, x_0) \shbb (y_0, y_1) \notag
\shbb (x_2, x_0) \chbb (y_0, y_2)\\
& v(x_0 - y_0) Q_{x_1 y_1} N_{x_2 y_2} \Omega\label{five}
 dx_1 \, dy_1 \, dx_2 \, dy_2 \, dx_0 \,dy_0~,
 \end{align}
 \begin{align}
 \int &\shbb (x_1, x_0) \shbb (y_0, y_1) \notag
\shbb (x_2, x_0) \shbb (y_0, y_2)\\
& v(x_0 - y_0) Q_{x_1 y_1} Q^*_{x_2 y_2} \Omega\label{six}
 dx_1 \, dy_1 \, dx_2 \, dy_2 \, dx_0 \,dy_0~.
 \end{align}
To compute the above six terms,
recall \eqref{QQ^*} through
 \eqref{NN} as well as \eqref{canonical}. In general,
  $N_{xy} \Omega = 1/2 \delta(x-y) \Omega$, while
  $\int f(x, y)Q^*_{xy} dx dy \Omega= (0, 0, f(x, y), 0 , \cdots)$
  up to symmetrization and normalization.

  The resulting contributions (neglecting symmetrization and normalization)
  follow.

  From \eqref{one}:
  \begin{align}
  &\psi(x_1, y_1)= \label{onee}\\
 &
\int \chbb (x_1, x_0) \chbb (y_0, y_1)\notag
\shbb (x_2, x_0) \chbb (y_0, x_2) v(x_0- y_0)\\
& \qquad \times dx_2 \,  dx_0 \, dy_0~.\notag
 \end{align}
 From \eqref{two}:
   \begin{align}
  &\psi(x_1, y_1, x_2, y_2)= \label{twoo}\\
 &
\int \chbb (x_1, x_0) \chbb (y_0, y_1)\notag
\shbb (x_2, x_0) \shbb (y_0, y_2) v(x_0- y_0)\\
  &\qquad \times dx_0 \, dy_0~.\notag
 \end{align}
 From \eqref{three}:
   \begin{align}
  &\psi= \label{threee}\\
 &
\int \shbb (x_1, x_0) \chbb (y_0, x_1)\notag
\shbb (x_2, x_0) \chbb (y_0, x_2) v(x_0- y_0)\\
&\qquad \times dx_1 \, dx_2 \,\notag
dx_0 \, dy_0~.
 \end{align}

 From \eqref{four}, with the $N$ and $Q^*$ reversed, we get

  \begin{align}
  &\psi(x_2, y_2)=\label{four1}\\
 &
\int \shbb (x_1, x_0) \chbb (y_0, x_1)\notag
\shbb (x_2, x_0) \shbb (y_0, y_2) v(x_0- y_0)\\
&\qquad \times  dx_1 \,  dx_0 \, dy_0~,\notag
 \end{align}
as well as the contribution from $[N, Q^*]$, i.e.
\begin{align}
  &\psi(y_1, y_2)= \label{four2}\\
 &
\int \shbb (x_1, x_0) \chbb (y_0, y_1)\notag
\shbb (x_1, x_0) \shbb (y_0, y_2) v(x_0- y_0)\\
& \qquad \times dx_1 \,  dx_0 \, dy_0~.\notag
 \end{align}
The contribution of \eqref{five} is
zero, and, finally, the contribution of \eqref{six}, using \eqref{QQ^*},
consists of four numbers, which can be represented by the two formulas
\begin{align}
  \psi&= \int \shbb (x_1, x_0) \shbb (y_0, x_1)
\shbb (x_2, x_0) \shbb (y_0, x_2) v(x_0- y_0)\label{six1}\\
&\qquad \times dx_1 \,  dx_2 \,dx_0 \, dy_0\notag
 \end{align}
 and
 \begin{equation}\label{six2}
  \psi=\int |\shbb |^2(x_1, x_0)
 |\shbb|^2 (y_0, y_1) v(x_0- y_0)
dx_1 \,  dy_1 \,dx_0 \, dy_0~.
 \end{equation}

We can now state the following proposition.

\begin{proposition} \label{vprop}
The state $e^B V e^{-B} \Omega $ has entries on the zeroth, second and fourth slot of a Fock space vector of the form given above.
In addition, if
\begin{align*}
\| \left( i
\frac{\partial}{\partial t} -\Delta_x -\Delta_y \right) \sh \|_{ L^1[0, T] L^2(dx dy)} \le C_1,
\end{align*}
\begin{align*}
\| \left( i
\frac{\partial}{\partial t} -\Delta_x +\Delta_y \right) p \|_{ L^1[0, T] L^2(dx dy)} \le C_2
\end{align*}
and
$v(x)=\frac{1}{|x|}$, or $v(x) = \chi(x)\frac{1}{|x|}$, then
\begin{align*}
\int _0^T\|e^B V e^{-B} \Omega\|^2_{\F} \, \, dt\le C~,
\end{align*}
where $C$ only depends on $C_1$ and $C_2$.
\end{proposition}

\begin{proof}
This follows by writing $\ch = \delta(x-y) + p$
and applying Cauchy-Schwartz and local smoothing estimates as in
the work of
Sj\"olin \cite{sj},  Vega \cite{vega}; see also  Constantin and Saut \cite{Const-Saut}.
In fact, we need the following slight generalization (see
Lemma \ref{sjl} below):
If
\begin{align*}
\| \left( i
\frac{\partial}{\partial t} -\Delta_{x_1} -\Delta_{x_2}
 \pm\Delta_{x_3} \cdots \pm\Delta_{x_n}
 \right) f(t, x_1, \cdots x_n) \|_{ L^1[0, T] L^2(dt dx)} \le C~,
\end{align*}
with initial conditions $0$,
then
\begin{align}
\|\frac{f(t, x_1, x_2, \cdots)}{|x_1-x_2|}\|_{ L^2[0, T] L^2(dx dy)}
\le C~. \label{sj}
\end{align}
We will check a typical term, \eqref{twoo}. This amounts to
proving the following three terms are in $L^2$.
\begin{enumerate}
\item
  \begin{align*}
  &\psi_{pp}(t, x_1, y_1, x_2, y_2)= \\
 &
\int p (t, x_1, x_0) p (t, y_0, y_1)\notag
\shbb (t, x_2, x_0) \shbb (t, y_0, y_2) v(x_0- y_0)\ dx_0 \, dy_0~.\notag
 \end{align*}
We use Cauchy-Schwartz in $x_0, y_0$ to get
\begin{align*}
&\int_0^T \int |\psi_{pp}|^2 dt\, dx_1\, dx_2\, d y_1\, dy_2\\
&\le \sup_t \int| p (t, x_1, x_0) p (t, y_0, y_1)|^2
dx_1\, dx_0\, d y_1\, dy_0\\
& \times \int_0^T \int
|\shbb (t, x_2, x_0) \shbb (t, y_0, y_2) v(x_0- y_0)|^2 dt\,
dx_2\, dx_0\, d y_2\, dy_0 \le C~.
\end{align*}
The first term is estimated by energy, and the second one is
an application of \eqref{sj} with $f = \shbb \shbb$. Notice that,
 because of the absolute value, we can
choose either $\sh$ or $\shb$ to insure that the Laplacians in $x_0$,
$y_0$ have the same signs.

\item
  \begin{align*}
  &\psi_{p\delta}(t, x_1, y_1, x_2, y_2)= \\
 &
\int p (t, x_1, x_0) \notag
\shbb (t, x_2, x_0) \shbb (t, y_1, y_2) v(x_0- y_1)\ dx_0~. \notag
 \end{align*}
Here, we use Cauchy-Schwartz in $x_0$ to estimate, in a similar fashion,
\begin{align*}
&\int_0^T \int |\psi_{p\delta}|^2 dt\, dx_1\, dx_2\, d y_1\, dy_2\\
&\le \sup_t \int| p (t, x_1, x_0) |^2
dx_1\, dx_0\\
& \times \int_0^T \int
|\shbb (t, x_2, x_0) \shbb (t, y_1, y_2) v(x_0- y_1)|^2 dt\,
dx_2\, dx_0\, d y_2\, dy_0 \le C~.
\end{align*}

\item
 \begin{align*}
  \psi_{\delta \delta}(x_1, y_1, x_2, y_2)=
\shbb (t, x_2, x_1) \shbb (t, y_1, y_2) v(x_1- y_1)~,
 \end{align*}which is just a direct application of \eqref{sj}.
\end{enumerate}
All other terms are similar.
\end{proof}
$\square$

We have to sketch the proof of the local smoothing estimate
that we used above.

\begin{lemma} \label{sjl}
If $f: \mathbb{R}^{3n+1} \to \mathbb{C}$ satisfies
\begin{align*}
\| \left( i
\frac{\partial}{\partial t} -\Delta_{x_1} -\Delta_{x_2}
 \pm\Delta_{x_3} \cdots \pm\Delta_{x_n}
 \right) f(t, x_1, \cdots x_n) \|_{ L^1[0, T] L^2(dx dy)} \le C
\end{align*}
with initial conditions $f(0, \cdots)=0$,
then
\begin{align*}
\|\frac{f(t, x_1, x_2, \cdots)}{|x_1-x_2|}\|_{ L^2[0, T] L^2( dx)}
\le C~.
\end{align*}
\end{lemma}

\begin{proof} We follow the general outline of Sjolin, \cite{sj}.
Using Duhamel's principle, it suffices to assume that
\begin{align}
\left( i
\frac{\partial}{\partial t} -\Delta_{x_1} -\Delta_{x_2}
 \pm\Delta_{x_3} \cdots \pm\Delta_{x_n}
 \right) f(t, x_1, \cdots x_n)=0 \label{eq}
\end{align}
with initial conditions $f(0, \cdots)= f_0 \in L^2$.
Furthermore, after the change of variables $ x_1 \to  \frac{x_1 + x_2}{\sqrt 2}$,
$x_2 \to \frac{ x_2 - x_1}{\sqrt 2}$, it suffices to prove
that
\begin{align*}
\|\frac{f(t, x_1, x_2, \cdots)}{|x_1|}\|_{ L^2[0, T] L^2( dx)}
\le C~,
\end{align*}
where $f$ satisfies  the same equation \eqref{eq}.
Changing notation, denote $x= (x_2, x_3, \cdots)$ and let $<\xi>^2$ be the relevant
expression $\pm |\xi_2|^2 \pm |\xi_3|^2 \ldots$. Write
\begin{align*}
f(t, x_1, x)= \int e^{i t (|\xi_1|^2 + <\xi>^2)} e^{i x_1 \cdot \xi_1 + i x \cdot \xi} \widehat{f}_0(\xi_1, \xi)\ d \xi_1\, d \xi~.
\end{align*}
Thus, we obtain
\begin{align*}
&\int \frac{|f(t, x_1, x)|^2}{|x_1|^2} dt dx_1 dx\\
&=\int
\int e^{i t (|\xi_1|^2- |\eta_1|^2 + <\xi>^2 - <\eta>^2)}
\frac{e^{i x_1 \cdot (\xi_1- \eta_1) + i x \cdot (\xi - \eta)}}{|x_1|^2} \widehat{f}_0(\xi_1, \xi)
\overline{\widehat{f}}_0(\eta_1, \eta)
 d \xi_1 \, d \xi \,d \eta_1\, d \eta\\
 &\qquad \times dt\, dx\, dx_1\\
 &=c \int \delta (|\xi_1|^2 - |\eta_1|^2) \frac{1}{|\xi_1- \eta_1|}\widehat{f}_0(\xi_1, \xi)
\overline{\widehat{f}}_0(\eta_1, \xi) d\xi_1 d\eta_1 d\xi\\
&\le \int| \widehat{f}_0(\xi_1, \xi)|^2 dx_1 \,d\xi~,
\end{align*}
because one can easily check that
\begin{align*}
\sup_{\xi_1} \int \delta (|\xi_1|^2 - |\eta_1|^2) \frac{1}{|\xi_1- \eta_1|} d \eta_1 \le C~.
\end{align*}
Thus, the kernel $\delta (|\xi_1|^2 - |\eta_1|^2) \frac{1}{|\xi_1- \eta_1|}$
is bounded from $L^2(d \eta_1)$ to  $L^2(d \xi_1)$.
\end{proof}
$\square$

\section{Error terms  $e^B[A, V]e^{-B}$}
\label{sec:error-II}

We proceed to check the operator $e^B[A, V]e^{-B}$.
The calculations of this section are similar to those of the preceding section
with the notable exception of \eqref{1two}--\eqref{1two2}.
Recall the calculations of Lemma \ref{Hartree-F} and write
\begin{align}
e^B[A, V]e^{-B}= \label{AV}
\int v(x-y) \Big( &\overline \phi (y)e^B a^*_xe^{-B} e^B a_x a_y
e^{-B}\\
& + \phi (y) e^B a^*_x a^*_y e^{-B} e^B a_x e^{-B} \Big)\ dx\, dy~. \label{AV1}
\end{align}
Now fix $x_0$. We start with the  term  \eqref{AV}.
According to Theorem \ref{algebra0}, we have
\begin{align*}
e^B a^*_{x_0}e^{-B} =\int \left(\sh (x, x_0) a_x + \chb (x, x_0) a^*_x \right) dx
\end{align*}
while $e^B a_{x_0} a_{y_0}e^{-B}$ has been computed in
\eqref{Q_{xy}}. The relevant terms are
\begin{align}
&\int \shbb (x, x_0) \chbb (y_0, y)\notag
 N_{x y} dx \, dy  \qquad \mbox{and}\\
&\int \shbb (x, x_0) \shbb (y_0, y)\notag
 Q^*_{x y} dx \, dy~.\\ \notag
\end{align}
Combining these two terms, there are three non-zero terms (which will act on $\Omega$):

\begin{enumerate}
\item
\begin{align}
&\int v(x_0-y_0)  \overline \phi (y_0)\label{1one}
\shbb (x_1, x_0)
 \shbb (x_2, x_0) \shbb (y_0, y_2)
 a_{x_1}Q^*_{x_2 y_2}
  \Omega\\
 &\qquad \times d x_1 dx_2 \, dy_2 \, dx_0 \, dy_0~.\notag
\end{align}
This term contributes terms of the form
\begin{align}
\psi(t, y_2)= \label{1one1}
\int v(x_0-y_0)  \overline \phi (t, y_0)
(\shbb (t, x_1, x_0))^2
  \shbb (t, y_0, y_2)
 d x_1\, dx_0 \, dy_0
\end{align}
as well as the term
\begin{align}
\psi(t, x_2)= \label{1one2}
&\int v(x_0-y_0)  \overline \phi (t, y_0)
\shbb (t, x_1, x_0)
 \shbb (t, x_2, x_0) \shbb (t, y_0, x_1)\\
 & \qquad\times d x_1 \, dx_0 \, dy_0~,\notag
\end{align}
which we know how to estimate. The second contribution is
\item
\begin{align}\label{1two}
&\int v(x_0-y_0)  \overline \phi (y_0)
\chbb (x_1, x_0)
 \shbb (x_2, x_0) \chbb (y_0, y_2)
 a^*_{x_1}
 N_{x_2 y_2} \Omega\\
 &\quad\times d x_1 dx_2 \, dy_2 \, dx_0 \, dy_0~.\notag
\end{align}
Commuting $a^*_{x_1}$ with $a_{x_2}$, we find that \eqref{1two} contributes
\begin{align}
\psi(t, y_2)=\label{1two1}
&\int v(x_0-y_0)  \overline \phi (t, y_0)
\chbb (t, x_1, x_0)
 \shbb (t, x_1, x_0) \chbb (t, y_0, y_2)\\
  &\qquad\times d x_1 \, dx_0 \, dy_0~. \notag
\end{align}
We expand $\chbb (t, x_1, x_0)=\delta(x_1-x_0) + \p(t, x_1-x_0)$.
The contributions of $p$ are similar to previous terms, but
$\delta(x_1-x_0)$ presents a new type of term, which will be addressed in
Lemma \ref{l2coll}. These contributions are
\begin{align}
\psi_{\delta p}(t, y_2)=\label{1two22}
&\int v(x_1-y_0)  \overline \phi (t, y_0)
 \shbb (t, x_1, x_1) \p (t, y_0, y_2)\\
  &  dx_1 \, dy_0 \notag
\end{align}
and
\begin{align}
\psi_{\delta \delta}(t, y_2)= \overline \phi (t, y_2) \label{1two2}
\int v(x_1-y_2)
 \shbb (t, x_1, x_1)
  d x_1~.
\end{align}

The last contribution of \eqref{AV} is
\item
\begin{align*}
&\int v(x_0-y_0)  \overline \phi (y_0)
\chbb (x_1, x_0)
 \shbb (x_2, x_0) \shbb (y_0, y_2)
 a^*_{x_1}Q^*_{x_2 y_2}
  \Omega\\
 &\qquad \times d x_1 dx_2 \, dy_2 \, dx_0 \, dy_0 \sim \psi(x_1, x_2, y_2)\notag\\
 \end{align*}
 where
 \begin{align*}
 &\psi(t, x_1, x_2, y_2)\\
 &=
 \int v(x_0-y_0)  \overline \phi (t, y_0)\notag
\chbb (t, x_1, x_0)
 \shbb (t, x_2, x_0) \shbb (t, y_0, y_2)
  dx_0 \, dy_0~,
\end{align*}modulo normalization and symmetrization. This term, as well as all the terms in \eqref{AV1}, are similar to previous ones and are omitted.
\end{enumerate}

We can now state the following proposition.

\begin{proposition}
The state $e^B [A, V] e^{-B} \Omega $ has entries in the first and third slot of a Fock space vector of the form given above.
In addition, if
\begin{align*}
\| \left( i
\frac{\partial}{\partial t} -\Delta_x -\Delta_y \right) \sh \|_{ L^1[0, T] L^2(dx dy)} \le C_1,
\end{align*}
\begin{align*}
\| \left( i
\frac{\partial}{\partial t} -\Delta_x +\Delta_y \right) p \|_{ L^1[0, T] L^2(dx dy)} \le C_1
\end{align*}
and
\begin{align}
\|\shbb(t, x, x)\|_{L^2([0, T] L^2(dx))} \le C_3~, \label{l2collapse}
\end{align}
and
$v(x) = \frac{\chi(x)}{|x|}$  for $\chi$ a $C_0^{\infty}$ cut-off function, then
\begin{align*}
\int_0^T\|e^B [A, V ]e^{-B} \Omega\|^2_{\F} \le C~,
\end{align*}
where $C$ only depends on $C_1$, $C_2$, $C_3$.
\end{proposition}
\begin{proof}
The proof is similar to that of Proposition \ref{vprop}, the only exception being the terms \eqref{1two22}, \eqref{1two2}.
It is only for the purpose of handling these terms that the Coulomb potential has to be truncated, since the convolution of the Coulomb potential with the $L^2$ function $\shbb(x, x)$ does not make sense. If $v$ is truncated to be in $L^1(dx)$, then we estimate the convolution  in $L^2(dx)$, and take
$\phi \in L^{\infty}(dy dt)$.
\end{proof}
$\square$

To apply this proposition, we need the following lemma.

\begin{lemma} \label{l2coll} Let $u \in X^{\frac{1}{2}, \frac{1}{2}+} $.
Then,
\begin{align*}
\|u(t, x, x)\|_{L^2 (dt \,dx)} \le C \|u\|_{X^{\frac{1}{2}, \frac{1}{2}+}}~.
\end{align*}
\end{lemma}
\begin{proof}
As it is well known, it suffices to prove the result for
$u$ satisfying
\begin{align*}
\left( i
\frac{\partial}{\partial t} -\Delta_x -\Delta_y \right) u(t, x, y) =0
\end{align*}
with initial conditions $u(0, x, y)=u_0(x, y) \in H^{\frac{1}{2}}$.
This can be proved as a ``Morawetz estimate", see~\cite{G-M},
or as a space-time estimate as in \cite{K-M}. Following the second approach,
the space-time Fourier transform of $u$ (evaluated at $2 \xi $ rather than $\xi$ for neatness) is
\begin{align*}
\widetilde{u}(\tau, 2 \xi)
&=c \int \delta(\tau - |\xi -\eta|^2 - |\xi +\eta|^2) \widetilde{u}_0(\xi- \eta, \xi+ \eta) d \eta\\
&=c \int \frac{\delta(\tau - |\xi -\eta|^2 - |\xi +\eta|^2)}
{(|\xi- \eta|+|\xi+ \eta|)^{1/2}}
 F(\xi- \eta, \xi+ \eta) d \eta~,\\
\end{align*}
where $F(\xi- \eta, \xi+ \eta)=(|\xi- \eta|+|\xi+ \eta|)^{1/2}
\widetilde{u}_0(\xi- \eta, \xi+ \eta)$.
By Plancherel's theorem, it suffices to show that
\begin{align*}
\|\widetilde{u}\|_{L^2(d \tau d \xi)} \le C \|F\|_{L^2(d \xi d \eta)}~.
\end{align*}
This, in turn, follows from the pointwise estimate (Cauchy-Schwartz with measures)
\begin{align*}
&|\widetilde{u}(\tau, 2 \xi)|^2\\
&\le c \int \frac{\delta(\tau - |\xi -\eta|^2 - |\xi +\eta|^2)}{
|\xi- \eta|+|\xi+ \eta|} d \eta\\
& \times\int\delta(\tau - |\xi -\eta|^2 - |\xi +\eta|^2)|F(\xi- \eta, \xi+ \eta)|^2
d\eta
\end{align*}
and the remark that
\begin{align*}
\int \frac{\delta(\tau - |\xi -\eta|^2 - |\xi +\eta|^2)}{
|\xi- \eta|+|\xi+ \eta|} d \eta \le C~.
\end{align*}
\end{proof}
$\square$

\section{The trace $\int d(t, x, x) dx$}
\label{sec:trace}
This section addresses the control of traces involved in derivations.
Recall that
 \begin{align*}
d(t, x, y)=
&
\left(i \sh_t +\sh g^T+ g \sh\right) \shb \notag\\
-&\left(i \ch_t + [g, \ch]\right)\ch\notag\\
-&\sh \overline m \ch -\ch m \shb~.
\end{align*}
Notice that if $k_1(x, y)\in L^2 (dx\, dy)$ and $k_2(x, y)\in L^2 (dx\, dy)$
then
\begin{align*}
&\int |k_1 k_2|(x, x) dx \le \int |k_1(x, y)| |k_2(y, x)| dy \, dx\\
&\le \|k_1\|_{L^2} \|k_2\|_{L^2}~.
\end{align*}
Recall from  Theorem \ref{localex} that if $v(x) = \frac{\epsilon}{|x|}$ or
$v(x) = \chi(x) \frac{\epsilon}{|x|}$ then $
i \sh_t +\sh g^T+ g \sh$, $i \ch_t + [g, \ch]$
and $\sh$ are in $L^{\infty}([0, 1])L^2(dx dy)$.
This allows us to control all traces except the contribution of $\delta(x-y)$ to the second term.
But, in fact, we have
\begin{align*}
 i\ch_t + [g, \ch]=
 \left(i k_t - \Delta_x k - \Delta_y k\right) \overline k
-k \overline{ \left(i k_t - \Delta_x k - \Delta_y k\right)} + \ldots~,
\end{align*}
which has bounded trace, uniformly in $[0, 1]$.

\end{document}